\documentclass[nofootinbib,twocolumn,pra,letterpaper]
{revtex4-2}
\usepackage{hyperref}
\usepackage{latexsym,amsmath,amssymb,amsfonts,graphicx,color,amsthm}
\usepackage{enumerate}
\usepackage{braket}
\usepackage{adjustbox}
\usepackage{footnote}
\usepackage[dvipsnames]{xcolor}
\usepackage{subcaption}

\newcommand{\cA}{\mathcal{A}}
\newcommand{\cB}{\mathcal{B}}
\newcommand{\cC}{\mathcal{C}}

\newcommand{\cE}{\mathcal{E}}
\newcommand{\cF}{\mathcal{F}}

\newcommand{\cH}{\mathcal{H}}
\newcommand{\cI}{\mathcal{I}}
\newcommand{\cJ}{\mathcal{J}}
\newcommand{\cK}{\mathcal{K}}
\newcommand{\cL}{\mathcal{L}}

\newcommand{\cP}{\mathcal{P}}
\newcommand{\cQ}{\mathcal{Q}}

\newcommand{\cS}{\mathcal{S}}
\newcommand{\cT}{\mathcal{T}}

\newcommand{\cW}{\mathcal{W}}
\newcommand{\cX}{\mathcal{X}}

\newcommand{\Id}{\mathbb{I}}
\newcommand{\tr}{\text{Tr}}

\newtheorem{theorem}{Theorem}
\newtheorem{proposition}{Proposition}
\newtheorem{lemma}{Lemma}
\newtheorem{corollary}{Corollary}

\newtheorem{example}{Example}

\newtheorem{remark}{Remark}
\newtheorem{conjecture}{Conjecture}

\begin{document}

\title{Quantifying unsharpness of observables in an outcome-independent way}
\author{Arindam Mitra$^{1,2}$}
\affiliation{$^1$Optics and Quantum Information Group, The Institute of Mathematical Sciences,
C. I. T. Campus, Taramani, Chennai 600113, India.\\
$^2$Homi Bhabha National Institute, Anushakti Nagar, Mumbai 400094, India.}

\date{\today}

\begin{abstract}

Recently, a very beautiful measure of the unsharpness (fuzziness) of the observables is discussed in the paper  [Phys. Rev. A \textbf{104}, 052227 (2021)]. The measure which is defined in this paper is constructed via uncertainty and does not depend on the values of the outcomes. There exist several properties of a set of observables (e.g., incompatibility, non-disturbance) that do not depend on the values of the outcomes. Therefore, the approach in the above-said paper is consistent with the above-mentioned fact and is able to measure the intrinsic unsharpness of the observables. In this work, we also quantify the unsharpness of observables in an outcome-independent way. But our approach is different than the approach of the above-said paper.  In this work, at first, we construct two Luder's instrument-based unsharpness measures and provide the tight upper bounds of those measures. Then we prove the monotonicity of the above-said measures under a class of fuzzifying processes (processes that make the observables more fuzzy). This is consistent with the resource-theoretic framework. Then we relate our approach to the approach of the above-said paper. Next, we try to construct two instrument-independent unsharpness measures. In particular, we define two instrument-independent unsharpness measures and provide the tight upper bounds of those measures and then we derive the condition for the monotonicity of those measures under a class of fuzzifying processes and prove the monotonicity for dichotomic qubit observables.  Then we show that for an unknown measurement, the values of all of these measures can be determined experimentally. Finally, we present the idea of the resource theory of the sharpness of the observables.

\end{abstract}

\maketitle
\section{Introduction}
In quantum mechanics, the observables are mainly of two types-(i) sharp observables and (ii) unsharp observables. Quantifying the unsharpness of observables is an interesting research direction to look at. Few works in this direction have been already done \cite{Busch-disturb, Ozawa-Uncer, Massar-Uncer, Busch-uncer-review, Busch-rms, Carmeli-approximate-repeat, Baek-entropic-uncer, Busch-book, Luo-u}. Recently, in the Ref. \cite{Luo-u}, the unsharpness of observables is quantified using uncertainty. The measure defined in the Ref. \cite{Luo-u}, is outcome-independent.

In this work, we also quantify the unsharpness of the observables in an outcome-independent way. But our approach is different than the approach of the Ref. \cite{Luo-u}. We first define two Luder's instrument-based measures. Then we discuss the different properties of these measures. Then, we try to construct two instrument-independent unsharpness measures. We provide a conjecture and if that can be proven, those instrument-independent measures will be consistent with the resource-theoretic framework for qubit observables. Then we discuss that the values of all of these measures can be determined experimentally. Then we make the justification for taking sharpness as a resource and also present the idea of the resource theory which can be completed in the future.

The rest of this paper is organized as follows. In Sec. \ref{sec:prelim}, we discuss the preliminaries.  From Sec. \ref{sec:L} we start discussing our main results. In particular, in Sec. \ref{subsec:L-c-b}, we construct two Luder's instrument-based unsharpness measures and provide the tight upper bounds of those measures. In Sec. \ref{subsec:el-mf}, we prove the monotonicity of the above-said measures under a class of fuzzifying processes. In Sec. \ref{subsec:F-EL}, we relate our approach to the approach of the Ref. \cite{Luo-u}. In the Sec. \ref{sec:E}, we try to construct two instrument-independent unsharpness measures. In particular, in Sec. \ref{subsec:E-c-b}, we define two instrument-independent unsharpness measures and provide the tight upper bounds of those measures. In Sec. \ref{subsec:e-mf}, we derive the condition for the monotonicity of those measures under a class of fuzzifying processes and prove the monotonicity for dichotomic qubit observables. In Sec. \ref{sec:ex}, we show that for an unknown measurement, the values of all of these measures can be determined experimentally. In Sec. \ref{sec:resource}, we present the idea of the resource theory of the sharpness of the observables. Finally, in Sec. \ref{sec:co}, we summarize our results and discuss the future outlook.

\section{Preliminaries}\label{sec:prelim}
In this section, we discuss the preliminaries.
\subsection{Observables}
An observable (positive operator-valued measure or POVM) $\cA$ acting on the Hilbert space $\cH$ is defined as a set of positive Hermitian matrices i.e., $\cA=\{A_i\}^n_{i=1}$ such that $\sum_iA_i=\Id_{d\times d}$ where $d$ is the dimension of the Hilbert space $\cH$ \cite{Nielsen-book, Teiko-book, Wilde-book}. The set $\{1,....,n\}$ is called outcome set of $\cA$ and is denoted by $\Omega_{\cA}$. Clearly $A_i\in\cL^+(\cH)$  and $\Id_{d\times d}\geq A_i \geq 0$ for all $i\in \Omega_{\cA}$ where $\cL^+(\cH)$ is the set of positive bounded linear operators acting on the Hilbert space $\cH$. Therefore, $A_i^2\leq A_i$ for all $i\in \Omega_{\cA}$.  If  $A^2_i=A_i$ holds for all $i\in\Omega_{\cA}$, we call $A$ as a projection-valued measure (PVM). PVMs are the sharp observable and clearly PVMs are the special cases of POVMs. Clearly, one outcome trivial sharp observable is $\Id_{d\times d}$. If there exist at least one $j\in\Omega_{\cA}$ such that $A^2_j<A_j$ then the observable $\cA$ is not a PVM. This type of observables are unsharp observables \cite{Luo-u}. 
\subsection{Quantum Channels}
A quantum channel $\Gamma:\cS(\cH)\rightarrow\cS(\cK)$ is a  completely positive trace preserving (CPTP) map where $\cS(\cH)$ is the state space (i.e., the set of density matrices on the Hilbert space $\cH$) \cite{Nielsen-book, Teiko-book}. For a quantum channel $\Gamma$, $\Gamma(\rho)$ can always be written as $\Gamma(\rho)=\sum_iK_i\rho K^{\dagger}_i$. This form of $\Gamma$ is called as the Kraus representation  of $\Gamma$ and $K_i$'s are called the Kraus operators of $\Gamma$. The  channel $\Gamma^*:\cL(\cK)\rightarrow\cL(\cH)$ is called as the dual channel (i.e., in Heisenberg picture) of $\Gamma:\cS(\cH)\rightarrow\cS(\cK)$ if for all $\rho\in\cS(\cH)$ and $X\in\cL(\cK)$, $\tr[\Gamma(\rho)X]=\tr[\rho\Gamma^*(X)]$ holds, where $\cL(\cH)$ is the set of bounded linear operators on the Hilbert space $\cH$.

A special type of channel is the depolarising channel. A depolarising channel $\Gamma^t_d$ is defined as $\Gamma^t_d(\rho)=t\rho+(1-t)\frac{\Id}{d}$ for all $\rho\in\cS(\cH)$ and $t\in[-\frac{1}{3}, 1]$. To simplify the notation we have written $\Id_{d\times d}$ as $\Id$.

\subsection{Quantum Instruments}
A quantum instrument $\cI$ is a set of completely positive (CP) maps $\{\Phi_i:\cS(\cH)\rightarrow\cL^+(\cK)\}$ i.e., $\cI=\{\Phi_i\}$ such that $\Phi=\sum_i\Phi_i$ is a quantum channel where  $\cL^+(\cH)$ is the set of positive bounded linear operators on the Hilbert space $\cH$ \cite{Teiko-book}. Suppose $\cA=\{A_i\}$ is an observable. A quantum instrument $\cI=\{\Phi_i\}$ is called $\cA$-compatible instrument if $\tr[\Phi_i(\rho)]=\tr[\rho A_i]$ for all $i\in \Omega_{\cA}$. Therefore, the observable $\cA$ can be measured using the instrument $\cI$.

There exist a special type of of instruments which are known as Luder's instruments. For an observable $\cA=\{A_i\}$, the $\cA$-compatible Luder's instrument is defined as $\cI^L_{\cA}=\{\Phi^{L}_{A_i}\}^{n_{\cA}}_{i=1}$ such that $\Phi^{L}_{A_i}(\rho)=\sqrt{A_i}\rho\sqrt{A_i}$ for all $i\in\{1,....,n_{\cA}\}$.

\subsection{Quantifying unsharpness of observables via uncertainty}\label{sub:prelim:uu}
In this subsection, we briefly discuss the approach of the Ref. \cite{Luo-u}. For the complete discussion, readers can check the Ref. \cite{Luo-u}. Suppose we have an observable $\cA=\{A_i\}^{n_{\cA}}_{i=1}$ acting on the Hilbert space $\cH$. Suppose the exact value of the $i$th outcome is $\alpha_i$. Let $\alpha$ be a row vector such that $\alpha=(\alpha_1, \alpha_2,....,\alpha_{n_{\cA}})$. Then $K^{\cA}_{\alpha}$ is defined as $K^{\cA}_{\alpha}=\sum_i\alpha_i A_i$. Similarly, $K^{\cA}_{\alpha^2}$ is defined as $K^{\cA}_{\alpha^2}=\sum_i\alpha^2_iA_i$. Next a noise operator $N^{\cA}_{\alpha}=K^{\cA}_{\alpha^2}-(K^{\cA}_{\alpha})^2$ is introduced. Then a function $F_{\rho}(\cA, \alpha)$ is introduced such that 

\begin{align}
F_{\rho}(\cA, \alpha)=\tr[\rho N^{\cA}_{\alpha}].
\end{align}
Clearly, $F_{\rho}(\cA, \alpha)\geq 0$. $F_{\rho}(\cA, \alpha)=0$ for all $\rho\in\cS(\cH)$ iff $\cA$ is a PVM. Now, it is shown in the Ref. \cite{Luo-u} that

\begin{align}
F_{\rho}(\cA, \alpha)=\alpha F_{\rho}(\cA)\alpha^T 
\end{align}

where $F_{\rho}(\cA)$ is a matrix such that

\begin{align}
[F_{\rho}(\cA)]_{ij}=\delta_{ij}\tr[\rho A_i]-\tr[\rho(\frac{A_iA_j+A_jA_i}{2})].\label{eq:F}
\end{align}

 Here $\delta_{ij}$ is Kronecker delta. Since, $F_{\rho}(\cA, \alpha)\geq 0$, $F_{\rho}(\cA)\geq 0$ and  $F_{\rho}(\cA)= 0$ iff $\cA$ is a PVM. \emph{This $F_{\rho}(\cA)$ matrix is independent of $\alpha$ and very important to construct the the unsharpness measure of an observable $\cA$.} Next, the matrix $\cF(\cA)$ is defined as $\cF(\cA)=F_{\frac{\Id}{d}}(\cA)$. Now it has been mentioned in the Ref. \cite{Luo-u} that \emph{any unitarily invariant norm} of $\cF(\cA)$ can quantify of the unsharpness of $\cA$. For simplicity they have taken $l^1$ norm $\|.\|_1$ which is defined as $\|X\|_1=\sum_{ij}|[X]_{ij}|$ for a matrix $X$. Therefore, the unsharpness measure of an observable $\cA$ is

\begin{align}
f(\cA)=\|\cF(\cA)\|_1.
\end{align}

\section{Luder's Instrument-based unsharpness measures of observables}\label{sec:L}
From this section, we start to discuss our main results.
\subsection{Construction and the upper bound of the Luder's Instrument-based unsharpness measure $\cE^L$}\label{subsec:L-c-b}
The sharp quantum observables (PVMs) have an interesting property that makes those observables different from the unsharp observables. Next, we discuss this property of PVMs which motivates us to quantify the unsharpness of the observables in the following outcome independent way. Suppose Alice is measuring an observable $\cA=\{A_i\}$ on a quantum state $\rho\in\cH_d$  through the Luder's instrument $\cL^{\cA}=\{\Phi^{L}_{A_i}(\rho)\}$. After obtaining the outcome $i$, the post measurement state will be $\rho^{\prime}_i=\frac{\Phi^{L}_{A_i}(\rho)}{\tr[\Phi^{L}_{A_i}(\rho)]}=\frac{\sqrt{A_i}_i\rho \sqrt{A_i}}{\tr[\rho A_i]}$. The probability of obtaining the outcome $i$ is $p_{i}=\tr[\rho A_i]$. Now, after obtaining the outcome $i$ if one more time the observable $\cA$ is measured by Alice on this post measurement state $\rho^{\prime}_i$, the probability of again obtaining the same outcome $i$ is

\begin{align}
p_{ii}=\tr[\rho^{\prime}_i A_i]=\frac{\tr[\rho A_i^2]}{\tr[\rho A_i]}.
\end{align}
  Now if $\cA$ is PVM $A^2_i=A_i$ for all $i$. Therefore, $p_{ii}=1$ if $\cA$ is a PVM.\emph{ Therefore, if $\cA$ is a PVM, on successive measurements of $\cA$, the outcome will definitely repeat.} If $\cA$ is not a PVM, there exist an outcome $j$ for which $A_j<A^2_j$ and therefore, $p_{jj}<1$. Therefore, there is a non-zero probability that an unsharp observable will not produce the same outcome on immediate successive measurements of the same observable. \emph{This is a feature of an unsharp observable or equivalently is the evidence of the unsharpness of an observable and is of course an essential difference between a PVM and POVM}. This fact motivates us to quantify the unsharpness of the observables in the following outcome-independent way.\\
In the above experiment, the average probability that any outcome will repeat in the successive measurement of $\cA$ is 

\begin{align}
\cP^L(\rho;\cA)&=\sum_i p_ip_{ii}=\sum_i\tr[\rho A^2_i]\nonumber\\
&=\tr[\rho \sum_iA^2_i]=\tr[\rho E^{\cA}]\label{eq:L-out-rep-E}
\end{align}
where $E^{\cA}=\sum_iA^2_i$. We will call $E^{\cA}$ as $E$-matrix of $\cA$. Clearly, $E^{\cA}$ is a positive Hermitian matrix and $E^{\cA}\leq \Id$. Now the average probability that a outcome will never repeat is 

\begin{align}
\cE^L(\rho;\cA)&=1-\cP^L(\rho;\cA)\nonumber\\
&=\tr[\rho(\Id-E^{\cA})]\label{eq:el:rho_A}\\
&\leq \|\rho\|_{tr}\|\Id-E^{\cA})\|\nonumber\\
&=\|(\Id-E^{\cA})\|\label{Luder-bound}
\end{align}

where $\|X\|$ denotes the operator norm i.e., the highest eigen value of a Hermitian matrix $X$ and $\|X\|_{tr}$ denotes the trace norm of a Hermitian matrix $X$ i.e., $\|X\|_{tr}=\tr[\sqrt{X^{\dagger}X}]$. In the second last inequality, we have used the fact that if $T\in\cL(\cH)$ is a trace-class (i.e., has a finite trace norm) Hermitian operator and $S\in\cL(\cH)$ is a arbitray Hermitian operator, then $\tr[ST]\leq \|T\|_{tr}\|S\|$ \cite{Teiko-book}. In the last equality, we have used the fact that $\rho$ is Hermitian and $\rho\geq 0$ and therefore, $\|\rho\|_{tr}=\tr[\rho]=1$. Now the bound written in equation \eqref{Luder-bound}, is achievable. Suppose $\ket{e^{\prime}_{max}}$ is the eigen state (i.e., normalised eigen vector) corresponding to the maximum eigen value of $(\Id-E^{\cA})$. Then, $\bra{e^{\prime}_{max}}E^{\cA}\ket{e^{\prime}_{max}}=\|(\Id-E^{\cA})\|$. Taking maximization of the quantity $\cE^L(\rho;\cA)$ over all set density matrices $\rho$, we obtain

\begin{align}
\cE^L(\cA)&=\max_{\rho}\cE^L(\rho;\cA)\nonumber\\
&=\bra{e^{\prime}_{max}}E^{\cA}\ket{e^{\prime}_{max}}\nonumber\\
&=\|(\Id-E^{\cA})\|.
\end{align}
We define $\cE^L(\cA)$ as the Luder's instrument-based unsharpness  measure of the observable $\cA$. Clearly, if $\cA$ is a PVM, $E^{\cA}=\Id$ and therefore, $\cE^L(\cA)=0$. If $\cA$ is not a PVM then there exists at least one $i$ such that $A_i^2<A_i$ and therefore, $E^{\cA}<\Id$ and therefore, $\cE^L(\cA)>0$. \emph{Therefore, $\cE^L$ is a faithful measure.} Clearly, $\cE^L$ measure is independent of the bijective relabeling of outcomes and of the values of outcomes.

There exist a upper bound for this unsharpness measure $\cE^L$. Our following lemma states that-

\begin{lemma}
For an observable $\cA=\{A_i\}^{n_{\cA}}_i$, $\cE^L(\cA)\leq 1-\frac{1}{n_{\cA}}$. This bound is achieved by the observable $\cT^{n_{\cA}}=\{I^{n_{\cA}}_i=\frac{\Id}{n_{\cA}}\}^{n_{\cA}}_{i=1}$.\label{le:elb}
\end{lemma} 

\begin{proof}
Suppose, $a^{\prime_{max}}=1-a_{min}$ is the maximum eigen value of $(\Id-E^{\cA})$ and $\ket{a^{\prime_{max}}}$ is the corresponding eigen vector. Threfore, $\|\Id-E^{\cA}\|=a^{\prime_{max}}$. This implies that $a_{min}$ is the minimum eigen value of $E^{\cA}$ and $\ket{a^{\prime_{max}}}=\ket{a_{min}}$ is the corresponding eigen vector. Now suppose, $\{\ket{n}\}$ is the eigen basis of $E^{\cA}$. Therefore, for some $n=n^{\prime}$, $\ket{n^{\prime}}=\ket{a_{min}}$. Then

\begin{align}
a_{min}&=\bra{a_{min}}E^{\cA}\ket{a_{min}}\nonumber\\
&=\bra{a_{min}}\sum_i A^2_i\ket{a_{min}}\nonumber\\
&=\sum_i\bra{a_{min}} A^2_i\ket{a_{min}}\nonumber\\
&=\sum_{i}\sum^d_{n=1}\bra{a_{min}} A_i\ket{n}\bra{n}E_i\ket{a_{min}}\nonumber\\
&=\sum_{i}\sum^d_{n=1}|\bra{a_{min}} A_i\ket{n}|^2\nonumber\\
&\geq\sum_{i}|\bra{a_{min}} A_i\ket{a_{min}}|^2\nonumber\\
&=\sum_{i} x_i^2\label{eq:min_x}
\end{align}
where $x_i=|\bra{a_{min}} A_i\ket{a_{min}}|$. Now we know that
\begin{equation}
\sum_i x_i=\sum_i|\bra{a_{min}} A_i\ket{a_{min}}|=1
\end{equation}
 as $A_i\geq 0$ for all $i\in\{1,....,n^{\cA}\}$ and $\sum_iA_i=\Id$. Now we know from the optimization method of Lagrange's undetermined multipliers that $\sum^n_{i=1}x^2_i$ takes the minimum value subject to condition $\sum^n_ix_i=1$ for $x_1=x_2=....=x_n=\frac{1}{n}$. Therefore, using this fact, the inequality \eqref{eq:min_x} becomes
 
\begin{align}
a_{min}&\geq \sum_ix_i^2\nonumber\\
&\geq \sum_i\frac{1}{n^2_{\cA}}\nonumber\\
&=\frac{1}{n_{\cA}}.
\end{align}

 This implies that
 
\begin{align}
\cE^L(\cA)&=\|\Id-E^{\cA}\|\nonumber\\
&=1-a_{min}\nonumber\\
&\leq 1-\frac{1}{n_{\cA}}.
\end{align}
 Now, for the observable $\cT^{n_{\cA}}$,
\begin{align}
\cE^L(\cT^{n_{\cA}})&=\|\Id-E^{\cT^{n_{\cA}}}\|\nonumber\\
&=\|\Id-\sum_i\frac{\Id}{n^2_{\cA}}\|\nonumber\\
&=\|\Id-\frac{\Id}{n_{\cA}}\|\nonumber\\
&=1-\frac{1}{n_{\cA}}.
\end{align}

 \end{proof}
 
 We can also define another measure of the unsharpness in a different way. It is to be noted that equation \eqref{eq:el:rho_A} is linear in $\rho$. Now let $\mathfrak{R}=\{\rho_1,....,\rho_k\}$ is a set of $k$ states.
 
 Then, the simple average (i.e., with same probability $\frac{1}{k}$) of $\cE^L(\rho;\cA)$ over this set $\mathfrak{R}$ is
 
 \begin{align}
 <\cE^L(\rho;\cA)>_{\mathfrak{R}}=&\sum^k_{i=1}\frac{1}{k}\cE^L(\rho_i;\cA)\nonumber\\
 =&\sum^k_{i=1}\frac{1}{k}\tr[\rho_i(\Id-E^{\cA})]\nonumber\\
 =&\tr[(\sum^k_{i=1}\frac{1}{k}\rho_i)(\Id-E^{\cA})]\nonumber\\
  =&\tr[<\rho>_{\mathfrak{R}}(\Id-E^{\cA})]\nonumber\\\label{eq:avg_r}
\end{align}

where $<\rho>_{\mathfrak{R}}=(\sum^k_{i=1}\frac{1}{k}\rho_i)$ is
the simple average (i.e., with same probability $\frac{1}{k}$) of the states over the set $\mathfrak{R}$. Generalising equation \eqref{eq:avg_r} for whole state space $\cS(\cH)$, we get that

\begin{align}
 <\cE^L(\rho;\cA)>_{\cS(\cH)}=&\tr[<\rho>_{\cS(\cH)}(\Id-E^{\cA})]\nonumber\\
  =&\tr[\frac{\Id}{d}(\Id-E^{\cA})]\nonumber\\
  =&\cE^L(\frac{\Id}{d};\cA).
\end{align}
where $<\rho>_{\cS(\cH)}$ is the simple average of the states over the whole state space $\cS(\cH)$ and in the second last equality we have used the well-known fact that $<\rho>_{\cS(\cH)}=\frac{\Id}{d}$. We define the unsharpness measure of $\cA$ as 

\begin{align}
\cE^{\prime L}(\cA)=&<\cE^L(\rho;\cA)>_{\cS(\cH)}\nonumber\\
&=\cE^L(\frac{\Id}{d};\cA).\label{eq:def_elprm}
\end{align}

Now the lemma below states the upper bound of $\cE^{\prime L}(\cA)$.

\begin{lemma}
For an observable $\cA=\{A_i\}^{n_{\cA}}_i$, $\cE^{\prime L}(\cA)\leq 1-\frac{1}{n_{\cA}}$. This bound is achieved by the observable $\cT^{n_{\cA}}=\{I^{n_{\cA}}_i=\frac{\Id}{n_{\cA}}\}^{n_{\cA}}_{i=1}$.\label{le:elprmb}
\end{lemma}

\begin{proof}
From the equation \eqref{eq:def_elprm}, we get

\begin{align}
\cE^{\prime L}(\cA)=&\cE^L(\frac{\Id}{d};\cA)\nonumber\\
\leq& \max_{\rho}\cE^L(\rho;\cA)\nonumber\\
=&\cE^L(\cA)\nonumber\\
\leq& 1-\frac{1}{n_{\cA}}.
\end{align}

Now, it is easy to check that $\cE^{\prime L}(\cT^{n_{\cA}})=1-\frac{1}{n_{\cA}}$.
\end{proof}

\begin{remark}\label{r1}
For an observable $\cA=\{A_i\}$, it is very easy to prove that $\cE^L(\cA)=\cE^L(\cA^U)$ and $\cE^{\prime L}(\cA)=\cE^{\prime L}(\cA^U)$ which $\cA^U=\{U^{\dagger}A_iU\}$. Therefore, $\cE^L$ and $\cE^{\prime L}(\cA)$ does not change if an unitary is acted on the observables in the Heisenberg picture.
\end{remark}

\subsection{Monotonicity of $\cE^L$ and $\cE^{\prime L}$ under a class of fuzzifying processes}\label{subsec:el-mf}
If $\cE^L$ is a useful measure of unsharpness (fuzziness), it should be monotonically non-decreasing under the processes which fuzzify the observables i.e., under the processes which make the observables more unsharp. These processes are called \emph{fuzzifying processes}. One may intuit that coarse-graining (a process where two or more outcomes are treated as a single one) is a fuzzifying process. We show through the next example that this is not true in general. 

\begin{example}\label{ex:postpross}
Consider two observables $\cA=\{A_i\}^3_{i=1}$ and $\cB=\{B_i\}^2_{i=1}$ acting on $\cH_3$ where $A_1=\frac{1}{2}\ket{1}\bra{1}+\frac{1}{4}\ket{2}\bra{2},~ A_2=\frac{1}{2}\ket{1}\bra{1}+\frac{3}{4}\ket{2}\bra{2},~ A_3=\ket{3}\bra{3}$ and $B_1=\ket{1}\bra{1}+\ket{2}\bra{2},~ B_2=\ket{3}\bra{3}$. clearly, $A_1+A_2=B_1$ and $A_3=B_2$ and therefore, $\cB$ is a coarse-graining of $\cA$. But $\cB$ is a PVM and $\cA$ is not a PVM. Therefore, $\cE^L(\cA)>0$ and $\cE^L(\cB)=0$. Therefore, under this kind of classical post-processing of the outcomes $\cE^L$ may be decreasing.
\end{example}

The above example shows that it is not possible to prove that $\cE^{\cL}$ is monotonically non-decreasing under the classical post-processing of outcomes as it is not a fuzzifying process, in general. Furthermore, one may intuit that the convex combination of observables is a fuzzifying process i.e. if an arbitrary observable $\cA$ is convexly combined with any other arbitrary observable, the resulting observable will be more unsharp than $\cA$. We will show through the next example that this is also not true, in general.

\begin{example}\label{ex:convex}
Consider a pair of observables $\cA=\{A_i\}$ and $\cB=\{B_i\}$ acting on $\cH_3$ where $A_1=\frac{1}{2}\ket{1}\bra{1}+\frac{1}{4}\ket{2}\bra{2},~ A_2=\frac{1}{2}\ket{1}\bra{1}+\frac{3}{4}\ket{2}\bra{2},~ A_3=\ket{3}\bra{3}$ and $B_1=\ket{1}\bra{1}, ~B_2=\ket{2}\bra{2},~B_3=\ket{3}\bra{3}$. We define a observable $\cC^{\lambda}=\{C^{\lambda}_i\}$ where $C^{\lambda}_i=\lambda A_i+(1-\lambda)B_i$ and $0\leq \lambda\leq 1$. Clearly, $C^{\lambda}_1=[1-\frac{\lambda}{2}]\ket{1}\bra{1}+\frac{\lambda}{4}\ket{2}\bra{2},~C^{\lambda}_2=\frac{\lambda}{2}\ket{1}\bra{1}+[1-\frac{\lambda}{4}]\ket{2}\bra{2},~C^{\lambda}_3=\ket{3}\bra{3}$. It can be observed that the sharpness of the observable $\cC^{\lambda}$ increases with  with the decrement of $\lambda$ and for $\lambda=0$, $\cC^{0}=\cB$ which is a PVM. Now since, for $\lambda=1$, $\cC^{1}=\cA$, $\cC^{\lambda}$ is always sharper than $\cA$ for all values of $\lambda$. It can be easily shown that $\cE^L(\cA)\geq \cE^L(\cC^{\lambda})$ for all values of $\lambda$'s.
\end{example}
The above example shows that it is also not possible to prove that $\cE^{\cL}$ is monotonically non-decreasing under the convex combination of observables as it is not a fuzzifying process, in general.\\
Example \ref{ex:postpross} and example \ref{ex:convex} suggest that it is not an easy task to specify all fuzzifying processes. But one can specify the special classes of fuzzifying processes. One can easily understand that the addition of white noise in the observables is a fuzzifying process. Therefore, we restrict ourselves to this particular class of fuzzifying processes and show that $\cE^{\cL}$ is monotonically non-decreasing under this class of fuzzifying processes in the following theorem.

\begin{theorem}
Suppose $\cA^{\lambda}=\{A^{\lambda}_i\}_{i=1}^{n_{\cA}}$ is an unsharp version of $\cA=\{A_i\}_{i=1}^{n_{\cA}}$ i.e., $A^{\lambda}_i=\lambda A_i+(1-\lambda)\frac{\Id}{n_{\cA}}$ for all $i\in\{1,.....,n_{\cA}\}$ where $1\geq\lambda\geq 0$. Then $\cE^L(\cA^{\lambda})\geq \cE^L(\cA)$ for all $1\geq\lambda\geq 0$.\label{th:el_mf}
\end{theorem}

 \begin{proof}
 The $E$-matrix of $\cA^{\lambda}$, is given by
 
 \begin{align}
 E^{\cA^{\lambda}}&=\sum_i(A^{\lambda}_i)^2\nonumber\\
 &=\sum_i(\lambda A_i+\frac{1-\lambda}{n_{\cA}}\Id)^2\nonumber\\
 &=\sum_i(\lambda^2 A_i^2+\frac{2\lambda(1-\lambda)}{n_{\cA}}A_i+\frac{(1-\lambda)^2}{n^2_{\cA}}\Id)\nonumber\\
 &=\lambda^2 \sum_iA_i^2+\frac{2\lambda(1-\lambda)}{n_{\cA}}\Id+\frac{(1-\lambda)^2}{n_{\cA}}\Id\nonumber\\
 &=\lambda^2 \sum_iA_i^2+\frac{(1-\lambda^2)}{n_{\cA}}\Id
 \end{align}
 Therefore, 
 \begin{align}
 \Id-E^{\cA^{\lambda}}=&\Id-(\lambda^2 \sum_iA_i^2+\frac{(1-\lambda^2)}{n_{\cA}}\Id)\nonumber\\
 =&\lambda^2(\Id-E^{\cA})+(1-\lambda^2)(1-\frac{1}{n_{\cA}})\Id.\label{eq:epm}
 \end{align}

Now, using the properties of the operator norm, we get

\begin{align}
 \|\Id-E^{\cA^{\lambda}}\|\leq&\|\lambda^2(\Id-E^{\cA})+(1-\lambda^2)(1-\frac{1}{n_{\cA}})\Id\|\nonumber\\
 &=\lambda^2\|\Id-E^{\cA}\|+(1-\lambda^2)(1-\frac{1}{n_{\cA}})\|\Id\|\nonumber\\
 &=\lambda^2\|\Id-E^{\cA}\|+(1-\lambda^2)(1-\frac{1}{n_{\cA}}).\label{eq:mod_epm}
 \end{align}

 Suppose, $\ket{e^{\prime}_{max}}$ is the eigen state (i.e., normalised eigenvector) of $(\Id-E^{\cA})$ corresponding to the highest eigenvalue of $(\Id-E^{\cA})$. Then  $\|\Id-E^{\cA}\|=\bra{e^{\prime}_{max}}(\Id-E^{\cA})\ket{e^{\prime}_{max}}$. Then, using the properties of the operator norm and equation \eqref{eq:epm}, we get
 
 \begin{align}
 \|\Id-E^{\cA^{\lambda}}\|\geq& \bra{e^{\prime}_{max}}(\Id-E^{\cA^{\lambda}})\ket{e^{\prime}_{max}}\nonumber\\
=&\lambda^2\bra{e^{\prime}_{max}}(\Id-E^{\cA})\ket{e^{\prime}_{max}}\nonumber\\
&+(1-\lambda^2)(1-\frac{1}{n_{\cA}})\bra{e^{\prime}_{max}}\Id\ket{e^{\prime}_{max}}\nonumber\\
&=\lambda^2\|(\Id-E^{\cA})\|+(1-\lambda^2)(1-\frac{1}{n_{\cA}}).\label{eq:mod_epm_i}
 \end{align}
 From inequality \eqref{eq:mod_epm} and inequality \eqref{eq:mod_epm_i}, we get
 
 \begin{equation}
 \|\Id-E^{\cA^{\lambda}}\|=\lambda^2\|(\Id-E^{\cA})\|+(1-\lambda^2)(1-\frac{1}{n_{\cA}}).
 \end{equation}
 
 Therefore,
 
 \begin{align}
 \cE^L(\cA^{\lambda})-\cE^L(\cA)=& (1-\lambda^2)(1-\frac{1}{n_{\cA}})-(1-\lambda^2)\|(\Id-E^{\cA})\|\nonumber\\
 =&(1-\lambda^2)[(1-\frac{1}{n_{\cA}})-\cE^L(\cA)]\nonumber\\
 \geq&0.
 \end{align}
  We have used Lemma \ref{le:elb} to obtain the last inequality. Hence the theorem is proved.
 \end{proof}
 
 Next we have an immediate corollary-
 
 \begin{corollary}
For any observable $\cA=\{A_i\}$, $\cE^L(\cA^{\lambda_2})\geq \cE^L(\cA^{\lambda_1})$ for all $1\geq\lambda_1\geq\lambda_2\geq 0$.\label{coro:el_mf}
 \end{corollary}
  
  \begin{proof}
  The observable $\cA^{\lambda_1}=\{A^{\lambda_1}_i=\lambda_1 A_i+(1-\lambda_1)\frac{\Id}{n_{\cA}}\}$. For notational simplicity we denote all $A^{\lambda_1}_i$ as $A^{\prime}_i$ i.e., $A^{\lambda_1}_i=A^{\prime}_i$ for all $i\in\{1,.....,n_{\cA}\}$ and we also denote the observable $\cA^{\lambda_1}$ as $\cA^{\prime}$ i.e., $\cA^{\lambda_1}=\cA^{\prime}$. Now the observable $\cA^{\lambda_2}=\{A^{\lambda_2}_i=\lambda_2 A_i+(1-\lambda_2)\frac{\Id}{n_{\cA}}\}$. Suppose $\gamma=\frac{\lambda_2}{\lambda_1}$. Clearly, $1\geq\gamma\geq 0$ as $\lambda_2\leq\lambda_1$ and both are positive. Then for all $i\in\{1,.....,n_{\cA}\}$,
  
  \begin{align}
  A^{\lambda_2}_i=&\lambda_2 A_i+(1-\lambda_2)\frac{\Id}{n_{\cA}}\nonumber\\
  =&\gamma\lambda_1A_i+(1-\gamma\lambda_1+\gamma-\gamma)\frac{\Id}{n_{\cA}}\nonumber\\
  =&\gamma\lambda_1A_i+[(1-\gamma)+\gamma(1-\lambda_1)]\frac{\Id}{n_{\cA}}\nonumber\\
  =&\gamma[\lambda_1A_i+(1-\lambda_1)\frac{\Id}{n_{\cA}}]+(1-\gamma)\frac{\Id}{n_{\cA}}\nonumber\\
  =&\gamma A^{\prime}_i+(1-\gamma)\frac{\Id}{n_{\cA}}\nonumber\\
  =&A^{\prime\gamma}_i\label{eq:al1-al2}
\end{align}   
where $A^{\prime\gamma}_i=\gamma A^{\prime}_i+(1-\gamma)\frac{\Id}{n_{\cA}}$ for all $i\in\Omega_{\cA^{\lambda_2}}$.
Therefore, $\cA^{\lambda_2}=\cA^{\prime\gamma}=\{A^{\prime\gamma}_i\}$. Then using the fact that $\cA^{\lambda_1}_i=\cA^{\prime}$ and Theorem \ref{th:el_mf}, we get that  $\cE^L(\cA^{\lambda_2})\geq \cE^L(\cA^{\lambda_1})$. Hence the corollary is proved.
  \end{proof}
 Therefore, $\cE^L(\cA^{\lambda})$ is monotonically non-decreasing with decreasing value of $\lambda$ or equivalently $\cE^L$ is monotonically non-increasing with increasing value of $\lambda$.
 
 Next, we have to prove the monotonicity of $\cE^{\prime L}$ under the addition of white noise. We start with our next theorem.
 
 \begin{theorem}
Suppose $\cA^{\lambda}=\{A^{\lambda}_i\}_{i=1}^{n_{\cA}}$ is an unsharp version of $\cA=\{A_i\}_{i=1}^{n_{\cA}}$ i.e., $A^{\lambda}_i=\lambda A_i+(1-\lambda)\frac{\Id}{n_{\cA}}$ for all $i\in\{1,.....,n_{\cA}\}$ where $1\geq\lambda\geq 0$. Then $\cE^{\prime L}(\cA^{\lambda})\geq \cE^{\prime L}(\cA)$ for all $1\geq\lambda\geq 0$.\label{th:elprm_mf}
\end{theorem}

\begin{proof}
 From equation \eqref{eq:def_elprm} and equation \eqref{eq:epm}, we get that
 \begin{align}
 \cE^{\prime L}(\cA^{\lambda})=&\cE^L(\frac{\Id}{d};\cA^{\lambda})\nonumber\\
 =&\tr[\frac{\Id}{d}(\Id-E^{\cA^{\lambda}})]\nonumber\\
 =&\tr[\frac{\Id}{d}(\lambda^2(\Id-E^{\cA})+(1-\lambda^2)(1-\frac{1}{n_{\cA}})\Id)]\nonumber\\
 =&\lambda^2\tr[\frac{\Id}{d}(\Id-E^{\cA})]+(1-\lambda^2)(1-\frac{1}{n_{\cA}})\nonumber\\
 =&\lambda^2\cE^{\prime L}(\cA)+(1-\lambda^2)(1-\frac{1}{n_{\cA}})
 \end{align}
 
 Therefore,
 \begin{align}
 \cE^{\prime L}(\cA^{\lambda})-\cE^{\prime L}(\cA)=&(\lambda^2-1)\cE^{\prime L}(\cA)+(1-\lambda^2)(1-\frac{1}{n_{\cA}})\nonumber\\
 =&(1-\lambda^2)[(1-\frac{1}{n_{\cA}})-\cE^{\prime L}(\cA)]\nonumber\\
 \geq& 0.
 \end{align}
 We have used Lemma \ref{le:elprmb} to obtain the last inequality. Hence, the theorem is proved.
\end{proof}
Next, we have an immdiate corollary

\begin{corollary}
For any observable $\cA=\{A_i\}$, $\cE^{\prime L}(\cA^{\lambda_2})\geq \cE^{\prime L}(\cA^{\lambda_1})$ for all $1\geq\lambda_1\geq\lambda_2\geq 0$.\label{coro:elprm_mf}
 \end{corollary}

\begin{proof}
The proof is similar to the proof of Corollary \ref{coro:el_mf}. From the equation \eqref{eq:al1-al2}, we get that $\cA^{\lambda_2}=\cA^{\prime\gamma}=\{A^{\prime\gamma}_i\}$. Then using the fact that $\cA^{\lambda_1}_i=\cA^{\prime}$ and Theorem \ref{th:elprm_mf}, we get that  $\cE^{\prime L}(\cA^{\lambda_2})\geq \cE^{\prime L}(\cA^{\lambda_1})$. Hence the corollary is proved.
\end{proof}

Therefore, $\cE^{\prime L}(\cA^{\lambda})$ is monotonically non-decreasing with decreasing value of $\lambda$ or equivalently $\cE^{\prime L}$ is monotonically non-increasing with increasing value of $\lambda$.

 \subsection{Relation of $\cE^L(\cA)$ and $\cE^{\prime L}(\cA)$ with $F_{\rho}(\cA)$}\label{subsec:F-EL}
 In this subsection, we relate the approach given in the Ref. \cite{Luo-u} (also briefly discussed in Sec. \ref{sub:prelim:uu}) with our apporach. More specifically, we relate $F_{\rho}(\cA)$ with $\cE^L(\cA)$ and $\cE^{\prime L}(\cA)$. From the expression of $F_{\rho}(\cA)$ i.e., from the equation \eqref{eq:F}, we get that
 
 \begin{align}
 \tr[F_{\rho}(\cA)]=&\sum_i[F_{\rho}(\cA)]_{ii}\nonumber\\
 =&\sum_i[\tr[\rho A_i]-\tr[\rho\frac{(A_iA_i+A_iA_i)}{2}]]\nonumber\\
 =&\sum_i[\tr[\rho A_i]-\tr[\rho A^2_i]]\nonumber\\
 =&\tr[\rho(\Id- \sum_iA^2_i)]\nonumber\\
 =&\tr[\rho (\Id-E^{\cA})]\nonumber\\
 =&\cE^L(\rho;\cA)\label{eq:f-el}
 \end{align}

Therefore,

\begin{align}
\cE^L(\cA)&=\max_{\rho} \cE^L(\rho;\cA)\nonumber\\
&=\max_{\rho}\tr[F_{\rho}(\cA)].
\end{align} 

Now as it is mentioned in the Ref. \cite{Luo-u} that $F_{\rho}(\cA)$ is Hermitian and $F_{\rho}(\cA)\geq 0$ for any arbitrary observable $\cA$, we have $\tr[F_{\rho}(\cA)]=\|F_{\rho}(\cA)\|_{tr}$.

\begin{align}
\cE^L(\cA)&=\max_{\rho}\|F_{\rho}(\cA)\|_{tr}.
\end{align}
 
\emph{Therefore, through our approach one of the operational meanings of the matrix $F_{\rho}(\cA)$ can be understood.}

Now taking $\rho=\frac{\Id}{d}$ and from equation \eqref{eq:f-el}, we get that

\begin{align}
\tr[F_{\frac{\Id}{d}}(\cA)]=&\cE^L(\frac{\Id}{d};\cA)\nonumber\\
=&\cE^{\prime L}(\cA).
\end{align}

As  $\tr[F_{\frac{\Id}{d}}(\cA)]=\|F_{\frac{\Id}{d}}(\cA)\|_{tr}$, we have

\begin{align}
\cE^{\prime L}(\cA)=\|F_{\frac{\Id}{d}}(\cA)\|_{tr}.
\end{align}

Now it has been mentioned in the Ref. \cite{Luo-u} that \emph{any unitarily invariant norm} of $\cF(\cA)=F_{\frac{\Id}{d}}(\cA)$ can quantify of the unsharpness of $\cA$. Therefore, we can take trace norm of $\cF(\cA)$ as a quantifier of the unsharpness of $\cA$ \cite{Chan}. Therefore, $\cE^{\prime L}$ measure is consistent with the Ref. \cite{Luo-u}.
 
\section{An attempt to construct instrument-independent unsharpness measures}\label{sec:E}
\subsection{Construction and the upper bound of the instrument-independent unsharpness measure $\cE$}\label{subsec:E-c-b}
In the previous section, we discussed two Luder's instrument-based unsharpness measures of observables. This discussion raises an immediate question: can one construct an instrument-independent unsharpness measure of observables? We try to answer this question in this section.

 Suppose Alice is using a general $\cA$-compatible quantum instrument $\cI^{\cA}=\{\Phi^{\cA}_i\}$ on the state $\rho$ to measure an observable $\cA=\{A_i\}$. Then, $q_i=\tr[\Phi^{\cA}_i(\rho)]=\tr[\rho A_i]$ is the probability of getting the outcome $i$ and $\rho^{\prime}_i=\frac{\Phi^{\cA}_i(\rho)}{\tr[\Phi^{\cA}_i(\rho)]}$ is the post-measurement after obtaining the outcome $i$. Now, after obtaining the outcome $i$ if one more time the observable $\cA$ is measured by Alice on this post measurement state $\rho^{\prime}_i$, the probability of again obtaining the same outcome $i$ is
 
\begin{align}
q_{ii}=\tr[\rho^{\prime}_iA_i].
\end{align}
The average probability that the outcome will repeat on successive measurements of the observable $\cA$ using the instrument $\cI^{\cA}$ on the state $\rho$ is 
\begin{align}
\cQ(\rho;\cA;\cI^{\cA})&=\sum_iq_iq_{ii}\nonumber\\
&=\sum_i\tr[\rho A_i]\tr[\rho^{\prime}_iA_i]\nonumber\\
&\leq\sum_i\tr[\rho A_i]\|A_i\|\nonumber\\
&=\tr[\rho \cX^{\cA}]
\end{align} 
where $\cX^{\cA}=\sum_i\|A_i\|A_i$. We will call $\cX^{\cA}$ as the $X$-matrix of the observable $\cA$. 

Therefore,

\begin{align}
\cQ(\rho;\cA)=\max_{\cI^{\cA}}\cQ(\rho;\cA;\cI^{\cA})
&\leq\tr[\rho \cX^{\cA}].\label{eq:xhigh}
\end{align}

Now the average probability that a outcome will never repeat is
\begin{align}
\cE(\rho;\cA;\cI^{\cA})&=1-\cQ(\rho;\cA;\cI^{\cA}).
\end{align}

 Now suppose, $a_{max}$ is the highest eigenvalue of the matrix $A_i$ and $\ket{a_{max}}$ is the corresponding eigen vector. Therefore, $\bra{a_{max}}A_i\ket{a_{max}}=\tr[\ket{a_{max}}\bra{a_{max}}A_i]=\|A_i\|$. Now consider an instrument $\cJ^{\cA}=\{\Theta^{\cA}_i\}$ where $\Theta^{\cA}_i(\rho)=\tr[\rho A_i]\ket{a_{max}}\bra{a_{max}}$. Therefore, the post-measurement states after obtaining the outcome $i$ is $\sigma_i=\ket{a_{max}}\bra{a_{max}}$. Now,

\begin{align}
\cQ(\rho;\cA;\cJ^{\cA})&=\sum_i\tr[\rho A_i]\tr[\ket{a_{max}}\bra{a_{max}}(A_i)]\nonumber\\
&=\sum_i\tr[\rho A_i]\|A_i\|\nonumber\\
&=\tr[\rho \cX^{\cA}].\label{eq:jeq}
\end{align}

Now,

\begin{align}
\cQ(\rho;\cA)&=\max_{\cI^{\cA}}\cQ(\rho;\cA;\cI^{\cA})\nonumber\\
&\geq\cQ(\rho;\cA;\cJ^{\cA})\nonumber\\
&=\tr[\rho \cX^{\cA}].\label{eq:xlow}
\end{align}

From  inequality \eqref{eq:xhigh}, equation \eqref{eq:jeq} and inequality \eqref{eq:xlow} we get

 \begin{align}
\cQ(\rho;\cA)=\tr[\rho \cX^{\cA}]=\cQ(\rho;\cA;\cJ^{\cA}).\label{eq:xq}
\end{align}
Therefore, choosing the best instrument $\cJ^{\cA}$, one can maximize the average probability that the outcome will repeat on successive measurements of the observable $\cA$ on the state $\rho$.

Now,

\begin{align}
\cE(\rho;\cA)&=\min_{\cI^{\cA}}\cE(\rho;\cA;\cI^{\cA})\nonumber\\
&=1-\max_{\cI^{\cA}}\cQ(\rho;\cA;\cI^{\cA})\nonumber\\
&=1-\tr[\rho\cX^{\cA}]\nonumber\\
&=\tr[\rho(\Id-\cX^{\cA})]\label{eq:e_rho_a}\\
&\leq\|\Id-\cX^{\cA}\|.
\end{align}

Therefore,

\begin{align}
\cE(\cA)&=\max_{\rho}\cE(\rho;\cA)\nonumber\\
&\leq\|\Id-\cX^{\cA}\|.\label{eq:elow}
\end{align}

Now suppose $x^{\prime}_{max}$ is the highest eigen value of $\Id-\cX^{\cA}$ and $\ket{x^{\prime}_{max}}$ is the corresponding eigen vector. Therefore, $\bra{x^{\prime}_{max}}\Id-\cX^{\cA}\ket{x^{\prime}_{max}}=\tr[\ket{x^{\prime}_{max}}\bra{x^{\prime}_{max}}(\Id-\cX^{\cA})]=\|\Id-\cX^{\cA}\|$. Then,

\begin{align}
\cE(\cA)&=\max_{\rho}\cE(\rho;\cA)\nonumber\\
&\geq \cE(\ket{x^{\prime}_{max}}\bra{x^{\prime}_{max}};\cA)\nonumber\\
&=\|\Id-\cX^{\cA}\|.\label{eq:ehigh}
\end{align}

From  inequality \eqref{eq:elow} and inequality \eqref{eq:ehigh} we get

\begin{align}
\cE(\cA)=\|\Id-\cX^{\cA}\|.\label{eq:e}
\end{align}

We define $\cE^{\cA}$ as the instrument-independent unsharpness measure of observables. Clearly, $\cE$ measure is independent of the bijective relabeling of outcomes and of the values of outcomes. If $\cA$ is a PVM, $\|A_i\|=1$ for all $i\in\{1,....,n_{\cA}\}$ and $\cE^{\cA}=0$. Next, we provide a remark on the faithfulness of $\cE$.

\begin{remark}
For an observable $\cA=\{A_i\}^{n_{\cA}}_{i=1}$ acting on the $d$-dimensional Hilbert space $\cH$, we know that $\cE(\cA)= 0$ only if
$\|A_i\|=1$ for all $i\in\Omega_A$. Let $\cA$ be an observable such that $\|A_i\|=1$ for all $i\in\Omega_{\cA}$ and $\ket{a^{max}_i}$ be the eigenstate (one of the eigen states if the maximum eigen value $1$ is degenerate) corresponding to the maximum eigen value $1$ for all $i\in\Omega_{\cA}$. Then for any two $i,j\in\Omega_{\cA}$ and $i\neq j$, suppose that $\braket{a^{max}_j|a^{max}_i}\neq 0$. Then $\bra{a^{max}_i}A_i+A_j\ket{a^{max}_i}\geq 1+\mid\braket{a^{max}_j|a^{max}_i}\mid^2> 1$. But Since, $A_i+A_j-\Id$, $\bra{\psi}A_i+A_j\ket{\psi}\leq 1$ for all $\ket{\psi}\bra{\psi}\in\cS(\cH)$. Hence, $\braket{a^{max}_j|a^{max}_i}= 0$ for all $i\in\Omega_{\cA}$. Now as $\sum_k A_k=\Id$ and $\bra{a^{max}_i}A_i\ket{a^{max}_i}=1$, we have $\bra{a^{max}_i}A_j\ket{a^{max}_i}=0$ for for any two $i,j\in\Omega_{\cA}$ and $i\neq j$. Therefore, for all $j\in\Omega_{\cA}$, there exist a $n_{\cA}-1$-dimensional subspace $\cK_{j}$ of $\cH$ such that for all $\ket{\psi}\bra{\psi}\in\cK_{j}$, $\bra{\psi}A_j\ket{\psi}=0$ and for $1$-dimensional subspace (i.e., for $\ket{a^{max}_j}$), $\bra{a^{max}_j}A_j\ket{a^{max}_j}=1$. Clearly, such construction is not possible for $n_{\cA}>d$ and for $n_{\cA}=d$, such construction implies $\cA=\{A_i=\ket{a^{max}_j}\bra{a^{max}_j}\}$ is a rank-$1$ PVM. Therefore, for $n_{\cA}\geq d$, $\cE(\cA)= 0$ implies $\cA$ is a PVM (sharp observable). Hence, the measure is faithful. \emph{ Above statement implies that $\cE$ is a faithful measure for all qubit observables (i.e., for $d=2$)}. For $n\leq d$, $\cE(\cA)> 0$ implise $\cA$ is an unsharp observable. But in this case $\cE(\cA)= 0$ does not implies $\cA$ is a PVM. For example- The qutrit observable $\cA^{\prime}=\{(\ket{1}\bra{1}+\frac{1}{2}\ket{2}\bra{2}), (\frac{1}{2}\ket{2}\bra{2}+\ket{3}\bra{3})\}$ is an unsharp observable. But $\cE(\cA^{\prime})=0$. \label{re:faith-E}
\end{remark}

Next we calculate the upper bound of $\cE$.

\begin{lemma}
For an observable $\cA=\{A_i\}^{n_{\cA}}_i$, $\cE(\cA)\leq 1-\frac{1}{n_{\cA}}$. This bound is achieved by the observable $\cT^{n_{\cA}}=\{I^{n_{\cA}}_i=\frac{\Id}{n_{\cA}}\}^{n_{\cA}}_{i=1}$.
\end{lemma}

\begin{proof}
From the definition of $\cE^{\rho, \cA}$, we have

\begin{align}
\cE(\rho; \cA)&=\min_{\cI^{\cA}}\cE^{\rho; \cA; \cI^{\cA}}\nonumber\\
&\leq \cE(\rho; \cA; \cL^{\cA})\label{eq:l_u}
\end{align}
Taking maximization over $\rho$ in both side of inequality \ref{eq:l_u} and from Lemma \ref{le:elb}, we get

\begin{align}
\cE^{\cA}\leq \cE^L(\cA)\leq 1-\frac{1}{n_{\cA}}.\label{eq:l<u_b}
\end{align}
Now for the observable  $\cT^{n_{\cA}}$,

\begin{align}
\cE(\cT^{n_{\cA}})&=\|\Id-\cX^{\cT^{n_{\cA}}}\|\nonumber\\
&=\|\Id-\sum_i\frac{\Id}{n^2_{\cA}}\|\nonumber\\
&=1-\frac{1}{n_{\cA}}
\end{align}
Hence, the lemma is proved.
\end{proof}

Similar to $\cE^{\prime L}$, we can define another instrument-independent unsharpness measure $\cE^{\prime}$ by taking average of $\cE(\rho;\cA)$ over full state space $\cS(\cH)$. Then

\begin{align}
\cE^{\prime}(\cA)=&<\cE(\rho;\cA)>_{\cS(\cH)}\nonumber\\
=&<\tr[\rho(\Id-\cX^{\cA})]>_{\cS(\cH)}\nonumber\\
=&\tr[<\rho>_{\cS(\cH)}(\Id-\cX^{\cA})]\nonumber\\
=&\tr[\frac{\Id}{d}(\Id-\cX^{\cA})]\nonumber\\
=&\cE(\frac{\Id}{d};\cA).\label{eq:def_eprm}
\end{align}

The statement similar to  Remark \ref{re:faith-E} also holds $\cE^{\prime}$.

Now the lemma below states the upper bound of $\cE^{\prime}(\cA)$.

\begin{lemma}
For an observable $\cA=\{A_i\}^{n_{\cA}}_i$, $\cE^{\prime}(\cA)\leq 1-\frac{1}{n_{\cA}}$. This bound is achieved by the observable $\cT^{n_{\cA}}=\{I^{n_{\cA}}_i=\frac{\Id}{n_{\cA}}\}^{n_{\cA}}_{i=1}$.\label{le:eprmb}
\end{lemma}

\begin{proof}
From the equation \eqref{eq:def_eprm}, we get

\begin{align}
\cE^{\prime}(\cA)=&\cE(\frac{\Id}{d};\cA)\nonumber\\
\leq& \max_{\rho}\cE(\rho;\cA)\nonumber\\
=&\cE(\cA)\nonumber\\
\leq& 1-\frac{1}{n_{\cA}}.
\end{align}

Now, it is easy to check that $\cE^{\prime}(\cT^{n_{\cA}})=1-\frac{1}{n_{\cA}}$.
\end{proof}

The statement similar to Remark \ref{r1} also holds for $\cE$ and $\cE^{\prime}$.
\subsection{Monotonicity of $\cE$ and $\cE^{\prime}$ under a class of fuzzifying processes}\label{subsec:e-mf}
Since from Example \ref{ex:postpross} and Example \ref{ex:convex}, we get that the coarse-graining and the convex combination of the observables are not the fuzzifying processes, monotonicity of $\cE$ can not be shown. Therefore, here we try to show that under the addition of white noise $\cE$ is monotonically non-decreasing. But unfortunately, it appears that the proof is not so straightforward. Therefore, at first, we derive the condition for the monotonicity of $\cE$ under the addition of white noise (i.e., Theorem \ref{th:con-E}). 

\begin{theorem} 
Suppose $\cA^{\lambda}=\{A^{\lambda}_i\}_{i=1}^{n_{\cA}}$ is an unsharp version of $\cA=\{A_i\}_{i=1}^{n_{\cA}}$ i.e., $A^{\lambda}_i=\lambda A_i+(1-\lambda)\frac{\Id}{n_{\cA}}$ for all $i\in\{1,.....,n_{\cA}\}$ where $1\geq\lambda\geq 0$. Then $\cE(\cA^{\lambda})\geq \cE(\cA)$ for all $1\geq\lambda\geq 0$ iff 

\begin{equation}
\Sigma^{\cA}_1\geq \Sigma^{\cA}_2\label{eq:con-E}
\end{equation}

 holds where  $\Sigma^{\cA}_1=(\bra{x^{\cA}_{min}}\cX^{\cA}\ket{x^{\cA}_{min}}-\frac{1}{n_{\cA}})$ and $\Sigma^{\cA}_2=(\frac{\sum_i\|A_i\|}{n_{\cA}}-\bra{x^{\cA}_{min}}\cX^{\cA}\ket{x^{\cA}_{min}})$ where $x^{\cA}_{min}$ is the lowest eigen value of $\cX^{\cA}$ and $\ket{x^{\cA}_{min}}$ is the eigen state of $\cX^{\cA}$ corresponding to the eigen value $x^{\cA}_{min}$.\label{th:con-E}
\end{theorem}

\begin{proof}
The $X$-matrix of $\cA^{\lambda}$ is 
\begin{align}
\cX^{\cA^{\lambda}}=&\sum_i\|\lambda A_i+(1-\lambda)\frac{\Id}{n_{\cA}}\|[\lambda A_i+(1-\lambda)\frac{\Id}{n_{\cA}}]\nonumber\\
=&\sum_i[\lambda \|A_i\|+(1-\lambda)\frac{1}{n_{\cA}}][\lambda A_i+(1-\lambda)\frac{\Id}{n_{\cA}}]\nonumber\\
=&\lambda^2\sum_i\|A_i\|A_i+\frac{\lambda(1-\lambda)}{n_{\cA}}[(\sum_i\|A_i\|)\Id+\sum_iA_i]\nonumber\\
&+\frac{(1-\lambda)^2}{n_{\cA}}\Id\nonumber\\
=&\lambda^2\cX^{\cA}+\frac{(1-\lambda)}{n_{\cA}}[\lambda(\sum_i\|A_i\|)+1]\Id
\end{align}

Therefore,
\begin{align}
\Id-\cX^{\cA^{\lambda}}=&\lambda^2[\Id-\cX^{\cA}]+(1-\lambda^2)\Id\nonumber\\
&-\frac{(1-\lambda)}{n_{\cA}}[\lambda(\sum_i\|A_i\|)+1]\Id\nonumber\\
=&\frac{(1-\lambda)}{n_{\cA}}[(n_{\cA}-1)+\lambda(n_{\cA}-\sum_i\|A_i\|)]\Id\nonumber\\
&+\lambda^2[\Id-\cX^{\cA}]\nonumber\\
=&\lambda^2[\Id-\cX^{\cA}]+\gamma\Id\label{eq:x_gamma}
\end{align}
where $\gamma=\gamma(\cA,\lambda)=\frac{(1-\lambda)}{n_{\cA}}[(n_{\cA}-1)+\lambda(n_{\cA}-\sum_i\|A_i\|)]$. As $A_i\leq\Id$ and therefore, $\sum_i\|A_i\|\leq n_{\cA}$, we have $\gamma\geq 0$. Therefore,

\begin{align}
\cE(\cA^{\lambda})=&\|\Id-\cX^{\cA^{\lambda}}\|\nonumber\\
=&\lambda^2\cE(\cA)+\gamma.
\end{align}

Now,
\begin{align}
\cE(\cA^{\lambda})-\cE(\cA)=&\gamma-(1-\lambda^2)\cE(\cA)\nonumber\\
=&\frac{(1-\lambda)}{n_{\cA}}[(n_{\cA}-1)+\lambda(n_{\cA}-\sum_i\|A_i\|)]\nonumber\\
&-(1-\lambda^2)\cE(\cA)\nonumber\\
=&(1-\lambda)[(1-\frac{1}{n_{\cA}}-\cE(\cA))\nonumber\\
&+\lambda(1-\frac{\sum_i\|A_i\|}{n_{\cA}}-\cE(\cA))]
\end{align}

Now, $\cE(\cA)=\|\Id-\cX^{\cA}\|=1-x^{\cA}_{min}=1-\bra{x^{\cA}_{min}}\cX^{\cA}\ket{x^{\cA}_{min}}$ where $x^{\cA}_{min}$ is the lowest eigen value of $\cX^{\cA}$ and $\ket{x^{\cA}_{min}}$ is the eigen state of $\cX^{\cA}$ corresponding to the eigen value $x^{\cA}_{min}$. Then

Therefore,
\begin{align}
\cE(\cA^{\lambda})-\cE(\cA)=&(1-\lambda)[(\bra{x^{\cA}_{min}}\cX^{\cA}\ket{x^{\cA}_{min}}-\frac{1}{n_{\cA}})\nonumber\\
&+\lambda(\bra{x^{\cA}_{min}}\cX^{\cA}\ket{x^{\cA}_{min}}-\frac{\sum_i\|A_i\|}{n_{\cA}})]\nonumber\\
&=(1-\lambda)[\Sigma^{\cA}_1-\lambda\Sigma^{\cA}_2]\nonumber\\
&=(1-\lambda)\Sigma^{\cA}(\lambda)
\end{align}
where $\Sigma^{\cA}_1=(\bra{x^{\cA}_{min}}\cX^{\cA}\ket{x^{\cA}_{min}}-\frac{1}{n_{\cA}})=x^{\cA}_{min}-\frac{1}{n_{\cA}}$, $\Sigma^{\cA}_2=(\frac{\sum_i\|A_i\|}{n_{\cA}}-\bra{x^{\cA}_{min}}\cX^{\cA}\ket{x^{\cA}_{min}})=\frac{\sum_i\|A_i\|}{n_{\cA}}-x^{\cA}_{min}$ and $\Sigma^{\cA}(\lambda)=[\Sigma^{\cA}_1-\lambda\Sigma^{\cA}_2]$.
Now, since $\cE(\cA)\leq(1-\frac{1}{n_{\cA}})$, $\Sigma^{\cA}_1\geq 0$. Now, There are two following cases-\\
\textbf{(I) For $\Sigma^{\cA}_2<0$ -}\\
In this case, $\cE(\cA^{\lambda})-\cE(\cA)\geq 0$ always. In this case $\Sigma^{\cA}_1\geq \Sigma^{\cA}_2$ trivially holds.\\
\textbf{(II) For $\Sigma^{\cA}_2\geq 0$ -}\\
In this case,  the minimum value of $\Sigma^{\cA}(\lambda)$ (for $\lambda=1$) is $\Sigma^{\cA}_{min}=[\Sigma^{\cA}_1-\Sigma^{\cA}_2]=2x^{\cA}_{min}-\frac{\sum_i\|A_i\|}{n_{\cA}}-\frac{1}{n_{\cA}}$. Clearly, the condition for $\cE(\cA^{\lambda})-\cE(\cA)\geq 0$ for all $\lambda$ is $\Sigma^{\cA}_{min}\geq 0$ or equivalently $\Sigma^{\cA}_1\geq \Sigma^{\cA}_2$.\\
\end{proof}

It appears that the proof of the inequality \eqref{eq:con-E} for arbitray observable acting on an arbitrary dimensional Hibert space, is difficult and therefore proof of the statement that under the addition of white noise $\cE$ is monotonically non-decreasing is difficult. Therefore, next we prove the inequality \eqref{eq:con-E} for the qubit dichotomic observables.

\begin{proposition}
For any dichotomic observable $\cW$, $\Sigma^{\cW}_1\geq \Sigma^{\cW}_2$ and therefore, $\cE(\cW^{\lambda})\geq \cE(\cW)$ for all $1 \geq\lambda\geq 0$.\label{th:qubit-di-E}
\end{proposition}

\begin{proof}
Suppose $\cW=\{W_1, W_2\}$ are two qubit dichotomic observables. Clearly $W_2=\Id-W_1$. Let $W_1=\omega_1\ket{\omega_1}\bra{\omega_1}+\omega_2\ket{\omega_2}\bra{\omega_2}$. Without the loss of generality, we can choose $\omega_1\geq \omega_2$. Then $\|W_1\|=\omega_1$. Now $W_2=(1-\omega_1)\ket{\omega_1}\bra{\omega_1}+(1-\omega_2)\ket{\omega_2}\bra{\omega_2}$. Clearly, $\|W_2\|=(1-\omega_2)$. Therefore,

\begin{align}
\cX^{\cW}=&\|W_1\|W_1+\|W_2\|W_2\nonumber\\
=&\omega_1[\omega_1\ket{\omega_1}\bra{\omega_1}+\omega_2\ket{\omega_2}\bra{\omega_2})]\nonumber\\
&+(1-\omega_2)[(1-\omega_1)\ket{\omega_1}\bra{\omega_1}+(1-\omega_2)\ket{\omega_2}\bra{\omega_2}]\nonumber\\
=&[\omega^2_1+(1-\omega_1)(1-\omega_2)]\ket{\omega_1}\bra{\omega_1}\nonumber\\
&+[\omega_1\omega_2+(1-\omega_2)^2]\ket{\omega_2}\bra{\omega_2}\nonumber\\
=&\omega^{\prime}_1\ket{\omega_1}\bra{\omega_1}+\omega^{\prime}_2\ket{\omega_2}\bra{\omega_2}\label{eq:wx}
\end{align}
where $\omega^{\prime}_1=[\omega^2_1+(1-\omega_1)(1-\omega_2)]$ and $\omega^{\prime}_2=[\omega_1\omega_2+(1-\omega_2)^2]$. Now 

\begin{align}
\omega^{\prime}_1-\omega^{\prime}_2=&[\omega^2_1+(1-\omega_1)(1-\omega_2)]-[\omega_1\omega_2+(1-\omega_2)^2]\nonumber\\
=&\omega^2_1+1+\omega_1\omega_2-\omega_1-\omega_2-\omega_1\omega_2-1+2\omega_2-\omega^2_2\nonumber\\
=&(\omega_1-\omega_2)[(\omega_1+\omega_2)-1].
\end{align}
Therefore, as $\omega_1\geq\omega_2$, we have $\omega^{\prime}_1\geq \omega^{\prime}_2$ for $(\omega_1+\omega_2)\geq 1$ and we have $\omega^{\prime}_1\leq \omega^{\prime}_2$ for $(\omega_1+\omega_2)\leq 1$.

Therefore, following two cases-\\
\textbf{(I) For $(\omega_1+\omega_2)\geq 1$ -}\\
In this case the minimum eigen value of $\cX^{\cW}$ is $x^{\cW}_{min}=\omega^{\prime}_2$. Therefore,

\begin{align}
\Sigma^{\cW}_{min}=&\Sigma^{\cW}_1-\Sigma^{\cW}_2\nonumber\\
=&2x^{\cW}_{min}-\frac{1}{2}-\frac{\|W_1\|+\|W_2\|}{2}\nonumber\\
=&2[\omega_1\omega_2+(1-\omega_2)^2]-\frac{1}{2}-\frac{\omega_1+(1-\omega_2)}{2}\nonumber\\
=&1-4\omega_2+2\omega^2_2+2\omega_1\omega_2-\frac{(\omega_1-\omega_2)}{2}.
\end{align}

Now from Fig. \ref{fig:sigma1}, we get that that $\Sigma^{\cW}_{min}\geq 0$ for all $\omega_1$ and $\omega_2$ satisfying the conditions $\omega_1\geq\omega_2$ and $\omega_1+\omega_2\geq 1$.

\textbf{(II) For $(\omega_1+\omega_2)\leq 1$ -}\\
In this case the minimum eigen value of $\cX^{\cW}$ is $x^{\cW}_{min}=\omega^{\prime}_1$. Therefore,

\begin{align}
\Sigma^{\cW}_{min}=&\Sigma^{\cW}_1-\Sigma^{\cW}_2\nonumber\\
=&2x^{\cW}_{min}-\frac{1}{2}-\frac{\|W_1\|+\|W_2\|}{2}\nonumber\\
=&2[\omega^2_1+(1-\omega_1)(1-\omega_2)]-\frac{1}{2}-\frac{\omega_1+(1-\omega_2)}{2}\nonumber\\
=&1+2\omega^2_1+2\omega_1\omega_2-2(\omega_1+\omega_2)-\frac{(\omega_1-\omega_2)}{2}.
\end{align}

Now from Fig. \ref{fig:sigma2}, we get that that $\Sigma^{\cW}_{min}\geq 0$ for all $\omega_1$ and $\omega_2$ satisfying the conditions $\omega_1\geq\omega_2$ and $\omega_1+\omega_2< 1$.

\begin{figure}
\begin{subfigure}{0.5\textwidth}
\includegraphics[width=7.9cm,height=8.1cm]{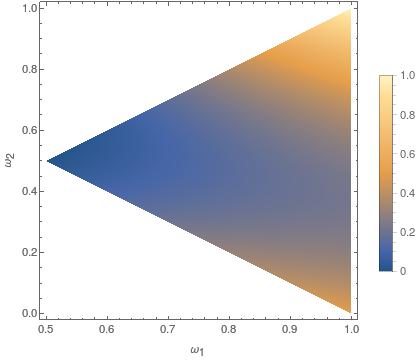}
\caption{Plot of $\Sigma^{\cW}_{min}$ w.r.t. $\omega_1$ and $\omega_2$ satisfying the conditions $\omega_1\geq\omega_2$ and $\omega_1+\omega_2\geq 1$ }\label{fig:sigma1}
\end{subfigure}
\begin{subfigure}{0.5\textwidth}
\includegraphics[width=7.9cm,height=8.1cm]{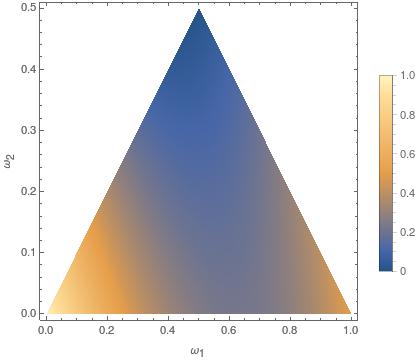}
\caption{Plot of $\Sigma^{\cW}_{min}$ w.r.t. $\omega_1$ and $\omega_2$ satisfying the conditions $\omega_1\geq\omega_2$ and $\omega_1+\omega_2< 1$}\label{fig:sigma2}
\end{subfigure}
\caption{Plots of $\Sigma^{\cW}_{min}$ w.r.t. $\omega_1$ and $\omega_2$ for $\omega_1\geq\omega_2$. These plots show that $\Sigma^{\cW}_{min}\geq 0$ always. }\label{aa}
\end{figure}

\end{proof}

Now, we have to prove monotonicity of $\cE^{\prime}$ under the addition of white noise. We start with our next theorem.

\begin{theorem} 
Suppose $\cA^{\lambda}=\{A^{\lambda}_i\}_{i=1}^{n_{\cA}}$ is an unsharp version of $\cA=\{A_i\}_{i=1}^{n_{\cA}}$ i.e., $A^{\lambda}_i=\lambda A_i+(1-\lambda)\frac{\Id}{n_{\cA}}$ for all $i\in\{1,.....,n_{\cA}\}$ where $1\geq\lambda\geq 0$. Then $\cE^{\prime}(\cA^{\lambda})\geq \cE^{\prime}(\cA)$ for all $1\geq\lambda\geq 0$ iff 

\begin{equation}
\Sigma^{\prime\cA}_1\geq \Sigma^{\prime\cA}_2\label{eq:con-Eprm}
\end{equation}

 holds where  $\Sigma^{\prime\cA}_1=\frac{1}{d}\tr[\cX^{\cA}]-\frac{1}{n}$ and $\Sigma^{\prime\cA}_2=\frac{\sum_i\|A_i\|)}{n_{\cA}}-\frac{1}{d}\tr[\cX^{\cA}]$.\label{th:con-Eprm}
\end{theorem}

\begin{proof}
From equation \eqref{eq:x_gamma}, we get that $(\Id-\cX^{\cA^{\lambda}})=\lambda^2[\Id-\cX^{\cA}]+\gamma\Id$ where $\gamma=\gamma(\cA,\lambda)=\frac{(1-\lambda)}{n_{\cA}}[(n_{\cA}-1)+\lambda(n_{\cA}-\sum_i\|A_i\|)]$. As $A_i\leq\Id$ and therefore, $\sum_i\|A_i\|\leq n_{\cA}$, we have $\gamma\geq 0$.  Therefore, from the equation \eqref{eq:def_eprm}, we get that 

\begin{align}
\cE^{\prime}(\cA^{\lambda})=\lambda^2\cE^{\prime}(\cA)+\gamma.
\end{align}

Therefore,

\begin{align}
\cE^{\prime}(\cA^{\lambda})-\cE^{\prime}(\cA)=&\gamma-(1-\lambda^2)\cE(\cA)\nonumber\\
=&(1-\lambda)[(1-\frac{1}{n_{\cA}}-\cE^{\prime}(\cA))\nonumber\\
&+\lambda(1-\frac{\sum_i\|A_i\|}{n_{\cA}}-\cE^{\prime}(\cA))]\nonumber\\
=&(1-\lambda)[\Sigma^{\prime\cA}_1-\lambda\Sigma^{\prime\cA}_2]\nonumber\\
=&\Sigma^{\prime\cA}(\lambda)
\end{align}
where $\Sigma^{\prime\cA}_1=(1-\frac{1}{n_{\cA}}-\cE^{\prime}(\cA))=(\frac{1}{d}\tr[\cX^{\cA}]-\frac{1}{n})$, $\Sigma^{\prime\cA}_2=(\frac{\sum_i\|A_i\|}{n_{\cA}}-\frac{1}{d}\tr[\cX^{\cA}])$ and $\Sigma^{\prime\cA}(\lambda)=[\Sigma^{\prime\cA}_1-\lambda\Sigma^{\prime\cA}_2]$. Now, since $\cE^{\prime}(\cA)\leq(1-\frac{1}{n_{\cA}})$, $\Sigma^{\prime\cA}_1\geq 0$. Now, There are two following cases-\\
\textbf{(I) For $\Sigma^{\prime\cA}_2<0$ -}\\
In this case, $\cE^{\prime}(\cA^{\lambda})-\cE^{\prime}(\cA)\geq 0$ always. In this case $\Sigma^{\cA}_1\geq \Sigma^{\cA}_2$ trivially holds.\\
\textbf{(II) For $\Sigma^{\prime\cA}_2\geq 0$ -}\\
In this case,  the minimum value of $\Sigma^{\prime\cA}(\lambda)$ (for $\lambda=1$) is $\Sigma^{\prime\cA}_{min}=[\Sigma^{\prime\cA}_1-\Sigma^{\prime\cA}_2]=2\frac{1}{d}\tr[\cX^{\cA}]-\frac{\sum_i\|A_i\|}{n_{\cA}}-\frac{1}{n_{\cA}}$. Clearly, the condition for $\cE^{\prime}(\cA^{\lambda})-\cE^{\prime}(\cA)\geq 0$ for all $\lambda$ is $\Sigma^{\prime\cA}_{min}\geq 0$ or equivalently $\Sigma^{\prime\cA}_1\geq \Sigma^{\prime\cA}_2$.\\

\end{proof}
Since, it is difficult to prove inequality \eqref{eq:con-Eprm}, we prove it for dichotomic qubit observables. Therefore, our next proposition is

\begin{proposition}
For any dichotomic observable $\cW$, $\Sigma^{\cW}_1\geq \Sigma^{\cW}_2$ and therefore, $\cE^{\prime}(\cW^{\lambda})\geq \cE^{\prime}(\cW)$ for all $1 \geq\lambda\geq 0$.\label{th:qubit-di-Eprm}
\end{proposition}

\begin{proof}
Suppose $\cW=\{W_1, W_2\}$ are two qubit dichotomic observables. Clearly $W_2=\Id-W_1$. Let $W_1=\omega_1\ket{\omega_1}\bra{\omega_1}+\omega_2\ket{\omega_2}\bra{\omega_2}$. Without the loss of generality, we can choose $\omega_1\geq \omega_2$. Then $\|W_1\|=\omega_1$. Now $W_2=(1-\omega_1)\ket{\omega_1}\bra{\omega_1}+(1-\omega_2)\ket{\omega_2}\bra{\omega_2}$. Clearly, $\|W_2\|=(1-\omega_2)$. Therefore, from equation \eqref{eq:wx}, we get that
\begin{align}
\cX^{\cW}=&\omega^{\prime}_1\ket{\omega_1}\bra{\omega_1}+\omega^{\prime}_2\ket{\omega_2}\bra{\omega_2}\label{eq:wxprm}
\end{align}
where $\omega^{\prime}_1=[\omega^2_1+(1-\omega_1)(1-\omega_2)]$ and $\omega^{\prime}_2=[\omega_1\omega_2+(1-\omega_2)^2]$. Therefore,
\begin{align}
\Sigma^{\prime \cW}_{min}&=2(\frac{1}{2}\tr[\cX^{\cW}])-\frac{1}{2}-\frac{\|W_1\|+\|W_2\|}{2}\nonumber\\
&=\omega^{\prime}_1+\omega^{\prime}_2-\frac{1}{2}-\frac{\omega_1+(1-\omega_2)}{2}\nonumber\\
&=\omega^2_1+\omega_1\omega_2+(1-\omega_2)(2-\omega_1-\omega_2)-1-\frac{\omega_1-\omega_2}{2}
\end{align}

\begin{figure}[hbt!]
\includegraphics[width=7.9cm,height=8.1cm]{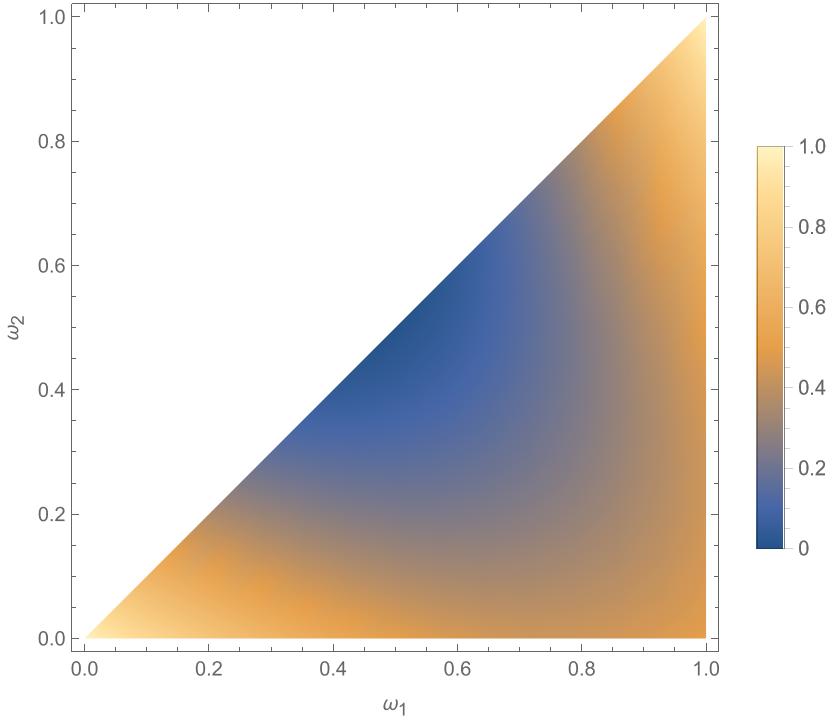}
\caption{Plot of $\Sigma^{\prime\cW}_{min}$ w.r.t. $\omega_1$ and $\omega_2$ for $\omega_1\geq\omega_2$. This plot show that $\Sigma^{\prime\cW}_{min}\geq 0$ always. }\label{plot:wprm}
\end{figure}
Figure \ref{plot:wprm}, says that  $\Sigma^{\prime\cW}_{min}\geq 0$ for $\omega_1\geq\omega_2$. Hence, $\cE^{\prime}(\cW^{\lambda})\geq \cE^{\prime}(\cW)$ for all $1 \geq\lambda\geq 0$.
\end{proof}

Therefore, inequality \eqref{eq:con-E} and inequality \eqref{eq:con-Eprm} hold for qubit dichotomic observables. We have searched for examples for which inequality \eqref{eq:con-E} inequality \eqref{eq:con-Eprm} do not hold. But we could not find any such example. Noting these facts, we provide the following conjecture-

\begin{conjecture}\label{co:emf}
For any qubit observable $\cA$, inequality $\Sigma^{\cA}_1\geq \Sigma^{\cA}_2$ and inequality $\Sigma^{\prime \cA}_1\geq \Sigma^{\prime \cA}_2$ hold  and therefore, $\cE(\cA^{\lambda})\geq \cE(\cA)$ and $\cE^{\prime}(\cA^{\lambda})\geq \cE^{\prime}(\cA)$ for all $1\geq\lambda\geq 0$.
\end{conjecture}

If Conjecture \ref{co:emf} can be proven then two corollaries similar to Corollary \ref{coro:el_mf} and Corollary \ref{coro:elprm_mf} can also be proven which establishes the monotonicity of $\cE$ and $\cE^{\prime}$ under the addition of white noise.

\section{Experimental Determination of the value of the unsharpness measures}\label{sec:ex}
Here we show that experimentally, one can determine the value of $\cE^L(\cA)$ and $\cE^{\prime L}(\cA)$ for an unknown qubit observable $\cA=\{A_i\}$. We show this for the qubit case. Generalization for the higher dimensions is straightforward.\\
Let $E$-matrix of an unknown qubit observable $\cA=\{A_i\}$ be 
$E^{\cA}=\begin{bmatrix}
a&c^*\\
c&d
\end{bmatrix}$ where this matrix is written in $\sigma_z$ basis. Suppose $\ket{\pm,i}$ are the eigen states of $\sigma_i$ corresponding to the eigen values $\pm 1$ for all $i\in\{x,y,z\}$. Suppose we have $n_{i,\pm}$ copies of such states are available to us. On each of these copies, $\cA$ has been measured twice successively using Luder's instrument. Suppose that for $f_{i,\pm}$ copies outcomes have repeated (i.e., the outcome of the first $\cA$ meaurement and the outcome of the second $\cA$ meaurement are same)  Then from equation \eqref{eq:L-out-rep-E}, we get that the average probability that any outcome will repeat, is

\begin{align}
\cP^L(\ket{\pm,i}\bra{\pm,i},\cA)=\tr[\ket{\pm,i}\bra{\pm,i}E^{\cA}].
\end{align}
Now we know that for large $n_{i,\pm}$, $\cP^L(\ket{\pm,i}\bra{\pm,i},\cA)\approx\frac{f_{i,\pm}}{n_{i,\pm}}$. Now $\cP^L(\ket{+,z}\bra{+,z},\cA)=a\approx\frac{f_{z,+}}{n_{z,+}}$, $\cP^L(\ket{-,z}\bra{-,z},\cA)=b\approx\frac{f_{z,-}}{n_{z,-}}$, $\cP^L(\ket{\pm,x}\bra{\pm,x},\cA)=\frac{a+b\pm2Re(c)}{2}\approx\frac{f_{x,\pm}}{n_{x,\pm}}$, $\cP^L(\ket{\pm,y}\bra{\pm,y},\cA)=\frac{a+b\pm2Im(c)}{2}\approx\frac{f_{y,\pm}}{n_{y,\pm}}$ where $Re(c)$ is the real part of $c$ and $Im(c)$ is the imaginary part of $c$. From these approximate equalities, we get the following set of approximate equalities-
\begin{align}
a\approx&\frac{f_{z,+}}{n_{z,+}};
~b\approx\frac{f_{z,-}}{n_{z,-}}\\
c\approx&(\frac{f_{x,+}}{2n_{x,+}}-\frac{f_{z,+}}{4n_{z,+}}-\frac{f_{z,-}}{4n_{z,-}})\nonumber\\
&+i(\frac{f_{y,+}}{2n_{y,+}}-\frac{f_{z,+}}{4n_{z,+}}-\frac{f_{z,-}}{4n_{z,-}}).
\end{align} 
Clearly for $n_{i\pm}\rightarrow\infty$, for all $i\in\{x,y,z\}$, above approximate equalities become exact equalities. In this way, if $a$, $b$ and $c$ are known approximately then $E^{\cA}$ is known approximately. The lowest eigenvalue of $E^{\cA}$ is $\frac{a+b-\sqrt{(a+b)^2-4(ab-|c|^2)}}{2}$. Therefore, $\cE^L(\cA)=\|\Id-E^{\cA}\|=1-\frac{a+b-\sqrt{(a+b)^2-4(ab-|c|^2)}}{2}$. Similarly, $\cE^{\prime L}(\cA)=1-\frac{1}{2}\tr[E^{\cA}]=1-\frac{a+b}{2}$.  Therefore, in this way, it is possible to determine the values of $\cE^L(\cA)$ and $\cE^{\prime L}(\cA)$ experimentally.

The experimental determination of the values of $\cE(\cA)$ and $\cE^{\prime}(\cA)$ is similar as above.
\section{An attempt to construct the resource theory of the sharpness of the observables }\label{sec:resource}
Quantification of quantum resources and the construction of the resource theory is very important and interesting direction of research \cite{Chitambar}. Few examples of different resource theories are (i) the resource theory of entanglement \cite{Chitambar,Shahandeh}, (ii) the resource theory of coherence \cite{Baumgratz, Andreas-Winter}, (iii) the resource theory of incompatibility \cite{Buscemi}, (iv) the resource theory of quantum channels \cite{Liu-re-cha}, (v) the resource theory of quantum thermodynamics \cite{Adesso-book, Goold-review} etc. We do not claim we construct the complete resource theory here. But we present the idea of the resource theory of the sharpness of the observables here. We take the sharpness of the observables as a resource here. We first provide the following reasons behind taking sharpness of the observables as a resource-
\begin{enumerate}
\item The Ref. \cite{Huber-proj} suggests that an ideal PVM have infinite resource costs. Therefore, with finite amount of resource, a PVM can not be performed with arbitrary accuracy. Therefore, this fact suggests that the ability to perform PVMs (i.e., sharp measurements) or equivalently sharpness of the observables itself can be considered as a resource.
\item In practice, it is very difficult to get rid of the interaction between the system and the environment. The interaction between the system and the environment disturbs the quantum state of a system or equivalently one can say that due to the interaction between the system and the environment, an effective channel $\Lambda$ acts on the system state. In Heisenberg picture, this channel acts on the observable $\cA$, which we want to measure, as $\Lambda^*(\cA)=\{\Lambda^*(A_i)\}$. Depending on the type of the interaction $\Lambda^*$ can convert a sharp observable into an unsharp observable. For an example- if $\Lambda=\Gamma^t_d$ is depolarising channel i.e., $\Lambda(\rho)=\Gamma^t_d(\rho)=t\rho+(1-t)\frac{\Id}{d}$ and $\cA=\{\ket{a_i}\bra{a_i}\}$ is a rank one PVM, then $\Lambda^*(\cA)=\Gamma^{t*}_d(\cA)=\{\Gamma^{t*}_d(A_i)=tA_i+(1-t)\frac{\Id}{d}\}$. Therefore, for a given value of $t<1$, it is impossible to perform a PVM accurately. Therefore, given the type of interaction, it may not be possible to perform a PVM with arbitrary accuracy. Therefore, to perform a PVM in a lab, one needs to make proper arrangements in the lab to get rid of such interactions between the system and the environment which prevents one to perform the desired PVM with arbitrary accuracy. Therefore, this fact also suggests that the ability to perform PVMs (i.e., sharp measurements) or equivalently sharpness of the observables itself can be considered as a resource.
\item There exist several information-theoretic tasks which can not be performed perfectly without the sharp observables. For example- a set of orthogonal states can be distinguished \emph{perfectly} only with certain PVMs. Therefore, this fact also suggests that the ability to perform PVMs (i.e., sharp measurements) or equivalently sharpness of the observables itself can be considered as a resource.
\end{enumerate}
Now we state the different elements of the resource theory of the sharpness of the observables below-

\begin{enumerate}
\item \emph{The resource-} The sharpness of the observables.
\item \emph{The free operation-} The fuzzifying processes. For example- a class of fuzzifying processes is the addition of white noise.
\item \emph{The resource measure-} We know that the unsharpness is opposite to the sharpness. Therefore, as sharpness is monotonically non-increasing under fuzzifying processes, the unsharpness is monotonically non-decreasing under fuzzifying processes. Since, from Theorem \ref{th:el_mf} and  Corollary \ref{coro:el_mf}, we get that $\cE^L$ is monotonically non-decreasing under the addition of white noise, $\cE^L$ can be a possible measure of unsharpness. The higher value of $\cE^L$ corresponds to less  sharpness (i.e., less resource). Similarly, from Theorem \ref{th:elprm_mf} and Corollary \ref{coro:elprm_mf}, we get that $\cE^{\prime L}$ can be a possible measure of unsharpness. It is to be noted that if the Conjecture \ref{co:emf} can be proven then $\cE$ and $\cE^{\prime}$ also can be an unsharpness measure for qubit observables consistent with the resource-theoretic framework.
\item \emph{Most resourceful measurements-} The sharp measurements (PVMs).

\item \emph{Free measurements-} Given the number of outcomes $n$,  the observable $\cT^{n}=\{I^{n}_i=\frac{\Id}{n}\}^{n}_{i=1}$ is a free measurement (most unsharp).

\item \emph{Example of an information-theoretic task which requires the resource-} Sharp measurements are required in the \emph{perfect} discrimination of the orthogonal states.
\end{enumerate}
Now a complete resource theory can be constructed only if all the fuzzifying processes are specified which is out of the scope of the present work. One point should be mentioned that the above-said resource theory is completely different the resource theory of quantum uncomplexity which is presented in the Ref. \cite{Faist} and the fuzzy operations which are discussed in the Ref. \cite{Faist} is quite different than our idea of fuzzifying processes.
\section{Conclusion}\label{sec:co}
In this work, at first, we have constructed two Luder's instrument-based unsharpness measures and provided the tight upper bounds of those measures. Then we have proved the monotonicity of the above-said measures under a class of fuzzifying processes (i.e., the addition of white noise). This is consistent with the resource-theoretic framework. We have also discussed the fact that these measures does not change if a unitary is acted on the observables in the Heisenberg picture. Then we have related our approach to the approach of the Ref. \cite{Luo-u}. Next, we have tried to construct tried instrument-independent unsharpness measures. In particular, we have defined two instrument-independent unsharpness measures and provided the tight upper bounds of those measures and then we have derived the condition for the monotonicity of those measures under a class of fuzzifying processes and proved the monotonicity for dichotomic qubit observables.  Then  we have shown that for an unknown measurement, the values of all of these measures can be determined experimentally. Finally, we have presented the idea of the resource theory of the sharpness of the observables.

It would be interesting to prove Conjecture \ref{co:emf} in the future. It would be also interesting to construct a complete resource theory of the sharpness of the observables in the future.

\section{Acknowledgements}
I would like to thank my advisor Prof. S. Ghosh for his valuable
comments on this work.

\end{document}